%% file: main.tex
\documentclass[aps,pra,preprint,superscriptaddress,longbibliography,nofootinbib]{revtex4-1}
\usepackage[utf8]{inputenc}
\usepackage{amsfonts}
\usepackage{amsmath}
\usepackage{amssymb}
\usepackage{amsthm}
\usepackage{mathtools}
\usepackage{enumerate}
\usepackage[shortlabels]{enumitem}
\usepackage{array}
\usepackage{bbm}
\usepackage{graphicx}
\usepackage{booktabs}
\usepackage{hyperref}
\usepackage{xcolor}
\usepackage{MnSymbol}

\newcommand\C{\mathbb{C}}
\newcommand\R{\mathbb{R}}

\newcommand\mc[1]{\mathcal{#1}}
\newcommand\cH{\mc{H}}
\newcommand\cK{\mc{K}}

\newcommand\cS{\mathcal{S}}
\newcommand\cA{\mc{A}}
\newcommand\cB{\mc{B}}

\newcommand\SAB{\cS(\cA \otimes \cB)}
\newcommand\cJ{\mc{J}}

\newcommand\cN{\mc{N}}
\newcommand\SN{\mc{S}(\mc{N})}

\newcommand\op{\mathrm{op}}

\newcommand\ad{\mathrm{ad}}
\newcommand\Ad{\mathrm{Ad}}

\newcommand\ra{\rightarrow}
\newcommand\mt{\mapsto}

\newcommand\tr{\mathrm{tr}}

\newtheorem{theorem}{Theorem}
\newtheorem{lemma}{Lemma}
\newtheorem{proposition}{Proposition}
\newtheorem{corollary}{Corollary}
\newtheorem{definition}{Definition}

\begin{document}

\title{Bipartite entanglement and the arrow of time}
%flag: Only separable states are timeless
\author{Markus Frembs}
\email{m.frembs@griffith.edu.au}
\affiliation{Centre for Quantum Dynamics, Griffith University,\\ Yugambeh Country, Gold Coast, QLD 4222, Australia}

\begin{abstract}
    We provide a new perspective on the close relationship between entanglement and time. Our main focus is on bipartite entanglement, where this connection is foreshadowed both in the positive partial transpose criterion due to Peres [A. Peres, Phys. Rev. Lett., \textbf{77}, 1413 (1996)] and in the classification of quantum within more general non-signalling bipartite correlations [M. Frembs and A. D\"oring, \url{http://arxiv.org/abs/2204.11471}]. Extracting the relevant common features, we identify a necessary and sufficient condition for bipartite entanglement in terms of a compatibility condition with respect to time orientations in local observable algebras, which express the dynamics in the respective subsystems. We discuss the relevance of the latter in the broader context of von Neumann algebras and the thermodynamical notion of time naturally arising within the latter.\\
\end{abstract}

\maketitle

%%%%%%%%%%%%%%%%%%%%%%%%%%%%%%%%%%%%%%%%%%%%%%%
\section{Introduction}\label{sec: introduction}
%Entanglement and time reversal - a common theme in disguise/collecting clues/an extended motivation

The connection between quantum entanglement and the arrow of time has been the subject of numerous research enterprises, some recent ones include \cite{PageWootters1983,JenningsRudolph2010,LinMarcolliOoguriStoica2015,Susskind2016,Racorean2019}.
%flag: reference (+ context?) to spacetime from entanglement approaches - AdS/CFT-correspondence, Caroll?, Ryu-Takayanagi formula etc.
Here, we mainly focus on bipartite entanglement. It has been surmised that the operation of partial transposition in the positive partial transpose criterion for bipartite entanglement due to Peres \cite{Peres1996} (henceforth referred to as the `PPT criterion'\footnote{Sometimes the criterion is also referred to as the Peres-Horodecki criterion \cite{Peres1996,Horodeckisz1996}})
%The criterion therefore more correctly describes separability.
is related to time reversal \cite{Horodeckisz2009}. Yet, this relationship seems not to have been made precise before. %For one obvious objection, the partial transpose is basis-dependent, whereas time reversal is a basis-independent, anti-unitary operator of the relevant $C^*$-algebra \cite{Wolf?}.
%Moreover, there is no immediate generalisation of the criterion to the multipartite case?

A related area of research, where time unexpectedly enters the picture, is the classification of quantum from non-signalling bipartite correlations. More precisely, the present author and Andreas D\"oring have recently shown that quantum states are characterised by a compatibility condition with respect to time orientations---roughly, the unitary evolution---in local observable algebras \cite{FrembsDoering2022a}.

We review the basics of and extract some key insights from these results in the following subsections. Building on those, in Sec.~\ref{sec: Generalised PPT criterion} we prove a necessary and sufficient criterion for bipartite entanglement. In Sec.~\ref{sec: entanglement and time-orientedness} we show that this, too, can be recast as a compatibility condition with respect to time orientations. Our work opens up various directions for future research, including practical considerations of our entanglement criterion as well as its generalisation to von Neumann algebras \cite{ConnesRovelli1994} (see Sec.~\ref{sec: time orientations and the arrow of time} and Sec.~\ref{sec: discussion}).

\input{EntanglementVSTime/PPTPureStates}
\input{EntanglementVSTime/NoSignallingCorrelations}

%%%%%%%%%%%%%%%%%%%%%%%%%%%%%%%%%%%%%%%%%%%%%%%
\section{Entanglement and the arrow of time}\label{sec: dilated PPT criterion}
%flag: use first section for PPT criterion only?!

In this main part, we combine the insights gained in previous sections.
%First, in Sec.~\ref{sec: Transposition vs conjugation} we relate the partial transpose to the ordering in operator (matrix) algebras and argue that it is more naturally understood in terms of the respective $*$-operation. We use this identification
In Sec.~\ref{sec: Generalised PPT criterion}, we identify a necessary and sufficient criterion for bipartite entanglement (Thm.~\ref{thm: dilated PPT criterion}).
%flag: In particular, we obtain a condition for the failure of the PPT criterion (in higher dimensions).
In Sec.~\ref{sec: entanglement and time-orientedness} we employ the structure theory of Jordan and (associative) $C^*$-algebras to translate this criterion into a compatibility condition between canonical time orientations on local subsystems (Thm.~\ref{thm: entanglement and time}). Finally, in Sec.~\ref{sec: time orientations and the arrow of time} we interpret our results in light of the intrinsic flow of time in von Neumann algebras by means of Tomita-Takesaki theory, Connes cocyles, and the background-independent thermodynamical arrow of time \cite{ConnesRovelli1994}.
%Altogether, we hope this puts the common (folklore) interpretation of the partial transpose operation in the PPT criterion as time reversal on a firm footing.

%\input{DilatedPPTCriterion/TranspositionVsConjugation}
\input{DilatedPPTCriterion/DilatedPPTCriterion}

%%%%%%%%%%%%%%%%%%%%%%%%%%%%%%%%%%%%%%%%%%%%%%%
\subsection{Entanglement and time orientations}\label{sec: entanglement and time-orientedness}

Comparing Prop.~\ref{prop: PPT trivial on decomposable states} with Thm.~\ref{thm: dilated PPT criterion}, it is natural to study the difference between Jordan $*$-homomorphisms and $C^*$-homomorphisms. To this end, we first review some basic facts about Jordan algebras and their dynamics in terms of one-parameter groups of automorphisms, before proving a reformulation of Thm.~\ref{thm: dilated PPT criterion} in terms of local time orientations.\\

\input{TheArrowOfTime/JordanAlgebras}
\input{TheArrowOfTime/Orientations}
\input{TheArrowOfTime/EntanglementAndTimeReversal}

\subsection{Time orientations and the arrow of time}\label{sec: time orientations and the arrow of time}

\input{TheArrowOfTime/TTTandCCC}
%flag: add to discussion, leads too far away; yet motivate as obvious generalisation of $S$ to arbitrary vN algebras (together with commutative operator model)

%%%%%%%%%%%%%%%%%%%%%%%%%%%%%%%%%%%%%%%%%%%%%%%
\section{Conclusion}\label{sec: discussion}

\input{Discussion/Conclusion}
%\input{Discussion/Outlook}
%\input{Discussion/OpenQuestions}

\paragraph*{Acknowledgements.} I thank Andreas D\"oring and Eric G. Cavalcanti for discussions. This work was supported by grant number FQXi-RFP-1807 from the Foundational Questions Institute and Fetzer Franklin Fund, a donor advised fund of Silicon Valley Community Foundation, and ARC Future Fellowship FT180100317.

%\clearpage
\bibliographystyle{siam}
\bibliography{bibliography}

\end{document}

%% file: EntanglementVSTime/PPTPureStates.tex
\subsection{The PPT criterion}\label{sec: PPT criterion}

Peres noted that the operation of partial transposition transforms separable states into separable states. In turn, any bipartite quantum state\footnote{We will identify states $\sigma \in \mc{S}(\cA)$ on $\cA$ with their respective density matrices via $\sigma(a) = \tr[\rho a]$ for all $a \in \cA$.} $\rho = \sum_{ij} c_{ij} \rho_{\cA,i} \otimes \rho_{\cB,j}
\in \SAB$ with $c_{ij} \in \C$, $\rho_{\cA,i} \in \mc{S}(\cA)$, $\rho_{\cB,j} \in \mc{S}(\cB)$, whose partial transpose $\rho^{T_\cA} := \sum_{ij} c_{ij} \rho^T_{\cA,i} \otimes \rho_{\cB,j}$ has at least one negative eigenvalue, is necessarily entangled \cite{Peres1996}. Throughout, we write $\cA = \mc{L}(\cH_\cA)$ for $\mathrm{dim}(\cH_\cA)$ finite. The partial-transpose criterion is necessary and sufficient in low dimensions, $\mathrm{dim}(\cH_\cA) = 2$ and $\mathrm{dim}(\cH_\cB) = 2,3$, but is merely sufficient in higher dimensions \cite{Horodeckisz1996}. Driven mainly by practical considerations %---in the wake of the new field of quantum information theory and its technological promises---
the result has been sharpened in various ways (see \cite{Horodeckisz2009} and references therein). This development, while rich and still active, has overshadowed the physical significance of Peres' insight. In contrast, here we will only be concerned with the conceptual importance, leaving the practical value of our work for future study.\\
%review: PPT criterion; dual motivation for entanglement classification (--> entanglement in mixed states for noisy, real-world quantum tech (cf. \cite{Horodeckisz1996} - only partly addressed); --> structural insight into mismatch separable vs HVM states (cf. \cite{Peres1996} - perspective taken here)); relation with distillation (\cite{Horodeckisz1998}) and time!; motivation problem of time!

\textbf{Pure and purified mixed states.} We recall the following simple fact.

\begin{proposition}\label{prop: PPT - pure bipartite case}
    The PPT criterion is necessary and sufficient for pure bipartite states. %$\rho = |\psi\rangle\langle\psi| \in \SAB$.
\end{proposition}

\begin{proof}
    Let $|\psi\rangle \in \cH_\cA \otimes \cH_\cB$ and consider the Schmidt decomposition $|\psi\rangle = \sum_i \alpha_i |ii\rangle$, $\alpha_i \in \R_+$, with density matrix $\rho_\psi = |\psi\rangle\langle\psi| = \sum_{ij} \alpha_i\alpha_j |i\rangle\langle j| \otimes |i\rangle\langle j| \in \mc{S}(\cA \otimes \cB)$. Note that $\rho_\psi$ is separable if and only if the sum in the Schmidt decomposition collapses to a single term. Applying partial transposition on system $\cA$ we obtain $\rho^{T_\cA}_\psi = \sum_{ij} \alpha_i\alpha_j |j\rangle\langle i| \otimes |i\rangle\langle j|$. It is easy to see that this operator has a negative eigenvalue for every pair of non-zero coefficients $\alpha_i,\alpha_j \neq 0$.
    %eigenvector: $(0,1,-1,0)$, eigenvalue: $-1$
\end{proof}
%NB: Note also that $\rho_\psi^T$ is pure if and only if $\rho_\psi$ is separable?!

%Clearly, this lifts to a necessary and sufficient condition in terms of the map $\phi_\psi$. The correspondence in Prop.~\ref{prop: PPT for CP maps}. Conversely, we may use the correspondence
%\begin{equation*}
    %\rho = (\rho_{ij}) = \sum_{ij} E_{ij} \otimes \rho_{ij}
    %= \sum_{ij} E_{ij} \otimes \phi_\rho(E_{ij})
    %= \sum_{ij} E_{ij} \otimes v^*\Phi(E_{ij})v
    %= \sum_{ij} E_{ij} \otimes \tilde{\rho}_{ij}
    %= (\tilde{\rho}_{ij}) = \tilde{\rho} \; ,
%\end{equation*}
%to obtain a state $\tilde{\rho} \in \cS(\cH_A \otimes \cK_B)$. With respect to this state we obtain a sharp classification by Lm.~\ref{lm: transposition vs local conjugation}.

%\textbf{Question:} Is Prop.~\ref{prop: PPT for CP maps} more general than Prop.~\ref{prop: PPT - pure bipartite case}? - pure states correspond with pure channels under the Choi-Jamio\l kowski isomorphism, i.e., channels with a single term in their Kraus representation; however, $\tilde{\rho}$ does not necessarily have to be pure?!\\

This is of course well known. The following observation is equally straightforward: we can apply Prop.~\ref{prop: PPT - pure bipartite case} to any bipartite state by considering purifications. What is more, a version of the Schr\"odinger-HJW theorem assures independence of the choice of purification \cite{Schroedinger1936,HughstonJoszaWootters1993} (see also \cite{Kirkpatrick2003}).
The PPT criterion thus becomes necessary and sufficient with respect to purifications. Of course, for pure states there are easier ways to check whether a state is entangled or separable.\footnote{For instance, note that we have used the Schmidt rank in Prop.\ref{prop: PPT - pure bipartite case}.} Nevertheless, as we will see below the criterion works \emph{because} it works on the level of purifications. This is best expressed in terms of channels.\\ %One motivation for this change of perspective comes from the action of partial transposition.\\

\textbf{Transposition vs Hermitian adjoint.} %Next, %we translate the (basis-dependent) operation of transposition into something more canonical.
%we consider the operation of transposition in more detail.
Note that transposition reverses the order of matrix multiplication: let $a \in M_{k\times l}(\C)$, $b \in M_{l\times m}(\C)$, then
\begin{equation*}\label{eq: transposition reverses matrix composition}
    ((ab)_{ki})^T = (ab)_{ik}
    = \sum_{j=1}^l a_{ij}b_{jk}
    = \sum_{j=1}^l (b_{kj})^T(a_{ji})^T
    = (b^Ta^T)_{ki}\; .
\end{equation*}
%This simple observation suggests two things. First, in order to account for the change in operator ordering it is natural to broaden the notion of algebraic structure from $C^*$- to Jordan algebras. We will take this up in Sec.~\ref{sec: entanglement and time-orientedness}. Second,
This fact is somewhat left implicit from the perspective of bipartite states $\rho \in \SAB$. In order to make it explicit, we identify a bipartite state $\rho$ with its quantum channel $\phi_\rho: \cA \ra \cB$ under the Choi-Jamio\l kowski isomorphim \cite{Jamiolkowski1972,Choi1975} (see also, \cite{Frembs2022b}): recall that every quantum channel $\phi: \cA \rightarrow \cB$, i.e., every completely positive %trace-preserving
linear map, determines a bipartite state $\rho_\phi$ (up to normalisation) by
\begin{equation}\label{eq: reverse CJ-isomorphism}
    \rho_\phi = \sum_{ij} E_{ij} \otimes \phi(E_{ij})\; ,\footnote{Sometimes this is written as $\rho_\phi = (\mathrm{id} \otimes \phi)(|\Phi\rangle\langle\Phi|)$, where $|\Phi\rangle = \sum_i |i\rangle \otimes |i\rangle$ is a maximally mixed state.}
\end{equation}
where $E_{ij}$ is the matrix with $1$ in the entry $(i,j)$ and $0$ elsewhere.

Conversely, every bipartite state $\rho \in \SAB$ corresponds to the quantum channel
\begin{equation}\label{eq: CJ-isomorphism}
    \phi_\rho(a) = \tr_{\cH_\cA} [\rho(a^T \otimes 1_\cB)]\; .\footnote{There are different versions of this isomorphism; for a detailed discussion, see \cite{Frembs2022b}.} %In particular, note that Eq.~(\ref{eq: CJ-isomorphism}) is basis-dependent. A basis-independent version, $\rho_\phi := \sum_{ij} E_{ji} \otimes \phi(E_{ij})$, is used in \cite{Jamiolkowski1972} (see also \cite{FrembsBuscemi2022}).
\end{equation}
%NB: argument needs to be independent of this choice
%NB: basis-independence hinges on the canonical identification between bases $\{|j\rangle\}_j$ and dual bases $\{\langle i|\}_i$.\footnote{The canonical dual basis $\{\langle i|\}_i$ to a basis $\{|j\rangle\}_j$ is given via the defining relation $\langle i|j\rangle = \delta_{ij}$.}
%NB: The disadvantage of the latter formulation is that completely positive maps $\phi_\rho$ no longer correspond with (positive) bipartite states $\rho_\phi$ under Choi's theorem. For this reason we will work with Eq.~(\ref{eq: reverse CJ-isomorphism}). We will use Choi's theorem in combination with Stinespring's theorem, by which the completely positive map $\phi_\rho$ has a Stinespring dilation $(\Phi,v,\cK)$, where $v: \cH_B \ra \cK$ and $\Phi: \cA \ra \cB(\cK)$ is a $C^*$-homomorphism such that $\phi_\rho = v^*\Phi v$.}
%Since $M_n(\cA) \cong \cA \otimes M_n(\C)$, we can write $x_{ij} \in M_n(\cA)$ as a $x_{ij} = \sum_{ij} E_{ij} \otimes x_{ij}$, where $E_{ij}$ is the matrix with $1$ in the entry $(i,j)$ and $0$ elsewhere.

Clearly, with respect to the choice of basis in Eq.~(\ref{eq: CJ-isomorphism}) we have $(E_{ij})^T = E_{ij}^*$. This allows us to replace transposition with the (Hermitian) adjoint.
%$S: \cA \ra \cA$ defined by $Sa = a^*$.

\begin{lemma}\label{lm: transposition vs *}
    Let $\rho \in \SAB$, let $\phi_\rho$ be the map under the linear isomorphism in Eq.~(\ref{eq: CJ-isomorphism}), and let $(\Phi_\rho,v,\cK)$ be a Stinespring dilation of $\phi_\rho$, i.e., $\phi_\rho = v^*\Phi_\rho v$ with $v: \cH_\cB \ra \cK$ linear and $\Phi_\rho: \cA \ra \cB(\cK)$ a $C^*$-homomorphism \cite{Stinespring1955}. Then $\phi_{\rho^{T_\cA}} = \phi^*_\rho = v^*\Phi^*_\rho v$.\footnote{Notably, this is different than so-called co-positive maps, i.e., maps $\phi^T := T_\cB \circ \phi$, where $\phi: \cA \ra \cB$ is a completely positive map: by a similar computation we obtain $\phi_{\rho^{T_\cA}} = \phi^T_\rho$, yet, $\phi^T_\rho \neq v^*\Phi^T_\rho v$ in general.}\footnote{We remark that $\phi^* := * \circ \phi$ denotes the adjoint of the image of the channel $\phi$, not its (Heisenberg) dual.}
\end{lemma}

\begin{proof}
    We have
    \begin{align*}
        \sum_{ij} E_{ij} \otimes \phi_{\rho^{T_\cA}}(E_{ij})
        = \rho^{T_\cA}
        %&= \sum_{ij} T(E_{ij}) \otimes \phi_\rho(E_{ij})
        %= \sum_{ij} E_{ij}^* \otimes \phi_\rho(E_{ij})
        = (\rho^*)^{T_\cA}
        = \sum_{ij} E_{ij} \otimes \phi^*_\rho(E_{ij})
        = \sum_{ij} E_{ij} \otimes v^*\Phi^*_\rho(E_{ij})v\; ,
    \end{align*}
    where we used Eq.~(\ref{eq: reverse CJ-isomorphism}) in the first and third step and $\rho^* = \rho$ in the second.
    %flag: $\Phi$ is a $C^*$-homomorphism by Stinespring's theorem \cite{Stinespring1955}, hence, $\Phi^*(A) = \Phi(A^*)$.
\end{proof}

Partial transposition therefore assumes a more natural interpretation in terms of the adjoint operation on the local system $\cB$ (see also \cite{Frembs2022b}). In Sec.~\ref{sec: time orientations and the arrow of time} we will see that this encodes a difference between %the respective notions of complete positivity in the form of
time orientations on the system $\cB$.
%What is more, (unlike transposition?) the $*$-operation defines a complex structure on the Jordan algebra of Hermitian operators (see Sec.~\ref{sec: Entanglement and time reversal asymmetry} below), thus promoting it to a $C^*$-algebra. %\footnote{In particular, such a complex structure mixes the complex structure of Hermitian matrices with complex entries and the products on it. This aspect is missing from the partial transpose condition?!}
The latter also play a crucial role in selecting quantum from more general non-signalling bipartite correlations \cite{FrembsDoering2022a}.

%% file: EntanglementVSTime/NoSignallingCorrelations.tex
\subsection{Quantum from non-signalling correlations}\label{sec: quantum from non-signalling correlations}
%flag: marketing section?!

It is instructive to view the problem of entanglement classification from the broader perspective of classifying quantum from non-signalling correlations. In general, non-signalling distributions are far from being quantum \cite{PopescuRohrlich1994}. Considering product quantum observables, a Gleason-type argument restricts non-signalling bipartite correlations to normalised linear functionals that are positive on pure tensors (POPT), yet not necessarily positive \cite{KlayRandallFoulis1987,Wallach2002,BarnumEtAl2010,FrembsDoering2022b}. To further single out quantum correlations among the latter, various additional physical principles have been proposed, see e.g. \cite{Popescu2014}. While successful in some instances, none has been shown to recover the quantum state space in general \cite{NavascuesGuryanovaHobanAcin2015}.

Recently, a classification of quantum states from more general non-signalling bipartite correlations has been obtained: quantum states correspond with those correlations, which satisfy (i) an extension of the no-signalling principle to dilations, and (ii) a relative consistency condition between the canonical unitary evolution in the respective subsystems (see Def.~1, Thm.~2 and Def.~2, Thm.~3 in \cite{FrembsDoering2022a}). Correlations under (i) (but not necessarily (ii)) correspond with decomposable maps under the Choi-Jamio\l kowski isomorphism in Eq.~(\ref{eq: CJ-isomorphism}).\footnote{We remark that for $\mathrm{dim}(\cH_\cA) = 2$ and $\mathrm{dim}(\cH_\cB) = 2,3$ every positive map $\phi: \cA \ra \cB$ is decomposable \cite{Stormer1963,Woronowicz1976}, which implies necessity of the PPT criterion in those %but not in higher
dimensions \cite{Horodeckisz1996} (see also Thm.~\ref{thm: entanglement and time} below).}
%flag: since the dilation argument is redundant in this case
%\input{CPvsDMaps}
%Let $\rho \in \mc{S}(\cH_A \otimes \cH_B)$ be a bipartite state, i.e., $\rho: \cB(\cH_A \otimes \cH_B) \ra \C$ is a positive linear functional. From \cite{Choi1975}, it follows that the map $\phi_\rho: \cB(\cH_A) \ra \cB(\cH_B)$ defined by the relation $\tr[\phi^*_\rho(a)b] = \rho(a \otimes b)$ (see Eq.~(\ref{eq: CJ-isomorphism}) below) is completely positive. Importantly, $\phi_\gamma$ is completely positive \emph{with respect to} the algebra $\cB(\cH_B)$; we stress this by writing $\phi^+$ explicitly. This does not imply that $\phi_\gamma$ is completely positive also with respect to $\cB_-(\cH_B)$.\\

Recall that a linear map $\phi: \cA \rightarrow \cB(\mc{H})$ is decomposable if there exists a Hilbert space $\mc{K}$, a bounded linear operator $v: \cH \rightarrow \mc{K}$, and a Jordan $*$-homomorphism $\Phi$, i.e., $\Phi(aa'+a'a) = \Phi(a)\Phi(a')+\Phi(a')\Phi(a)$ and $\Phi^*(a) = \Phi(a^*)$ for all $a,a' \in \cA$ (for details, see Sec.~\ref{sec: entanglement and time-orientedness}), such that $\phi = v^* \Phi v$. Such maps are more general than quantum channels $\phi: \cA \ra \cB(\cH)$, which are of similar form: $\phi = v^*\Phi v$ with $\Phi$ a $C^*$-homomorphism. By Stinespring's theorem \cite{Stinespring1955}, the latter is equivalent to $\phi$ being completely positive: if $x_{ij} \in M_n(\cA)_+ = (M_n(\C) \otimes \cA)_+$, then $\phi(x_{ij}) := \mathrm{id}_{M_n(\C)} \otimes \phi(x_{ij}) \in M_n(\cB(\mc{H}))_+$. Similarly, decomposable maps can be characterised by a weaker positivity condition \cite{Stormer1982}: if $x_{ij} \in M_n(\cA)_+$ \emph{and} $x_{ji} \in M_n(\cA)_+$, then $\phi(x_{ij}) \in M_n(\cB(\mc{H}))_+$. Let $\mc{S}_D(\cA \otimes \cB)$ denote the class of bipartite states corresponding to decomposable maps under the Choi-Jamio\l kowski isomorphism. Interestingly, (the weaker positivity condition in) $\mc{S}_D(\cA \otimes \cB)$ is preserved under partial transposition (see also \cite{Frembs2022b}).

\begin{proposition}\label{prop: PPT trivial on decomposable states}
    Let $\rho \in \mc{S}_D(\cA \otimes \cB)$, i.e., $\rho$ corresponds to a decomposable map under the Choi-Jamio\l kowski isomorphism in Eq.~(\ref{eq: reverse CJ-isomorphism}). Then $\rho^{T_\cA} \in \mc{S}_D(\cA \otimes \cB)$.
\end{proposition}

\begin{proof}
    By a similar argument to the one in Lm.~\ref{lm: transposition vs *}, $\phi_{\rho^{T_\cA}} = \phi^*_\rho = v^*\Phi^*_\rho v$, where $\Phi_\rho$ is a Jordan $*$-homomorphism, hence, $\Phi^*_\rho := * \circ \Phi_\rho = \Phi_\rho \circ *$. But then so is $\Phi^*_\rho$: for all $a_1,a_2 \in \cA$,
    \begin{equation*}
        \Phi^*_\rho(\{a_1,a_2\})
        = \Phi_\rho(\{a_1,a_2\}^*)
        %= \Phi(\{a_2^*,a_1^*\})
        = \Phi_\rho(\{a_1^*,a_2^*\})
        = \{\Phi_\rho(a_1^*),\Phi_\rho(a_2^*)\}
        = \{\Phi^*_\rho(a_1),\Phi^*_\rho(a_2)\}\; .
    \end{equation*}
    Consequently, $\phi_{\rho^{T_\cA}}$ is decomposable and $\rho^{T_\cA} \in S_D(\cA\otimes\cB)$.
\end{proof}
%In particular, we have that for all $y_i \in \cA$:
%\begin{equation*}
    %\sum_{i,j=1}^n (\phi(x_{ij})x_i,x_j) = \sum_{i,j=1}^n (\Phi(x_{ij})vx_i,vx_j) = \sum_{i,k=1}^n (\Phi^+(x_{ij})vx_i,vx_j) + \sum_{i,k=1}^n (\Phi^-(x_{ij})vx_i,vx_j) \geq 0\; .
%\end{equation*}
%Cleary, the PPT criterion is no longer trivial in this case: $x_{ij} \in M_n(\cA)_+$ does not imply $t(x_{ij}) = x_{ji} \notin M_n(\cA)_+$ in general. Consequently, $\sum_{ij}(\Phi^-(x_{ij})x_i,x_j) = \sum_{ij} (\Phi^-(t(x_{ji}))x_i,x_j) \ngeq 0$ and we cannot conclude that $x_{ij} \in M_n(\cA)_+$ implies $\sum_{ij}^n (\phi(x_{ij})x_i,x_j) \geq 0$.

%NB: We can say more. The positivity condition inherent to the decomposable map $\phi = v^*\Phi v$ is a mixture of the positivity conditions with respect to different algebras. More precisely, every Jordan $*$-isomorphism $\Phi: \cA \ra \cB$ \mf{(with $\cA$ a factor)} is the sum of a $*$-isomorphism $\Phi'$ and a $*$-anti-isomorphism $\Phi''$ (Thm.~10, \cite{Kadison1951}).%\footnote{Note that this general fact does not imply the special fact that in low dimensions ($2\times2$ and $2\times3$) positive maps are decomposable \cite{Woronowicz1976}.} 
%Noting that an anti-$C^*$-homomorphism $\Phi: \cA \ra \cB$ is a $C^*$-homomorphism $\Phi: \cA^\mathrm{opp} \ra \cB$, we find that $\phi = \phi' + \phi''$, where $\phi' = v^*\Phi' v: \cA \ra \cB(\cH)$ and $\phi'' = v^*\Phi'' v: \cA^\mathrm{opp} \ra \cB(\cH)$ are completely positive maps by Stinespring's theorem \cite{Stinespring1955}.\\

Summarising this motivational prologue, in Sec.~\ref{sec: PPT criterion} we remarked that the PPT criterion becomes necessary and sufficient when applied to \emph{purifications}, and used the Choi-Jamio\l kowski isomorphism to translate the criterion \emph{from bipartite states to bipartite channels}. In particular, we recast partial transposition into the (Hermitian) adjoint in Lm.~\ref{lm: transposition vs *}, which further suggests to lift the PPT criterion from the level of bipartite quantum channels $\phi$ to $C^*$-homomorphisms $\Phi$ in Stinespring dilations $\phi = v^*\Phi v$.

More generally, in Sec.~\ref{sec: quantum from non-signalling correlations} we considered the PPT criterion with respect to dilations of decomposable maps. Prop.~\ref{prop: PPT trivial on decomposable states} shows that the partial transpose preserves the respective positivity condition of linear functionals corresponding to such maps. \emph{We deduce that the PPT criterion is sensitive precisely to the difference between decomposable and completely positive maps.} In the context of classifying quantum from non-signalling bipartite correlations, this is achieved by enforcing a compatibility condition with respect to the relative time orientation between systems $\cA$ and $\cB$ %(step (ii) in)
\cite{FrembsDoering2022a}. Building on this motivation, in the next section we identify a necessary and sufficient condition for bipartite separability in terms of a compatibility condition with respect to \emph{different time orientations} between $\cA$ and $\cB$.
%This will be done in two steps: first, we identify a necessary and sufficient condition for bipartite entanglement in Sec.~\ref{sec: dilated PPT criterion}, before we elaborate on the compatibility condition with respect to time orientations in Sec.~\ref{sec: entanglement and time-orientedness}. Sec.~\ref{sec: discussion} embeds our research into a wider context and highlights several directions for future research.

%% file: DilatedPPTCriterion/DilatedPPTCriterion.tex
\subsection{A necessary and sufficient criterion for bipartite entanglement}\label{sec: Generalised PPT criterion}

The PPT criterion translates between bipartite states and bipartite channels as follows:
\begin{align*}\label{eq: PPT under JC}
    \rho^{T_\cA} \mathrm{\ positive} &\stackrel{\mathrm{\ \ \ Choi's\ theorem}\ \ \ \ }{\Longleftrightarrow} & &\phi_{\rho^{T_\cA}} \mathrm{\ completely\ positive} \\
    &\stackrel{\ \ \ \ \ \ \ \ \mathrm{Lm.~}\ref{lm: transposition vs *}\ \ \ \ \ \ \ \ }{\Longleftrightarrow} & &\phi^*_\rho \mathrm{\ completely\ positive} \\
    &\stackrel{\mathrm{\ Stinespring's\ theorem}}{\Longleftrightarrow} & &\phi^*_\rho = (v')^*\Phi'_\rho v', \mathrm{\ where\ } (v',\Phi'_\rho,\cK') \mathrm{\ is\ a\ Stinespring\ dilation}
\end{align*}
Now since $\phi_\rho$ is completely positive, it also has a Stinespring dilation $\phi_\rho = v^*\Phi_\rho v$.
%Lm.~\ref{lm: transposition vs *} therefore raises the question under what condition $\Phi^*: \cA \ra \cB$ is $C^*$-homomorphism as well.
\emph{We may thus strengthen the PPT criterion as follows: rather than $\phi_\rho^*$ admitting any Stinespring dilation, we ask when $v' = v$, $\Phi'_\rho = \Phi^*_\rho$, i.e., when $\Phi^*_\rho$ is a $C^*$-homomorphism.} %This condition is stronger since separability not only implies that $T_1(\rho) > 0$, but also that $\Phi^*$ is a $C^*$-homomorphism. In fact, we will see that unlike the PPT criterion the latter criterion is also necessary.

Note first that this condition does not depend on the choice of Stinespring dilation.

\begin{lemma}\label{lm: Schroedinger-HJW for maps}
    Let $\phi = v_1\Phi_1 v_1^* = v_2\Phi_2 v_2^*$ be two Stinespring dilations of $\phi$. Then $\Phi_1^*: \cA \ra \cB(\cK_1)$ is a $C^*$-homomorphism if and only if $\Phi_2^*: \cA \ra \cB(\cK_2)$ is a $C^*$-homomorphism.
\end{lemma}

\begin{proof}
    %Let $(\Phi_1,v_1,\cK_1)$ and $(\Phi_2,v_2,\cK_2)$ be two Stinespring dilations for $\phi_\rho = v_1^* \Phi_1 v_1 = v_2^* \Phi_2 v_2$.
    There is a partial isometry $W: \cK_1 \ra \cK_2$ defined by $W \Phi_1 v_1 |\psi\rangle = \Phi_2 v_2 |\psi\rangle$ for all $|\psi\rangle \in \cH_\cB$ such that $\Phi_1 = W^* \Phi_2 W$.\footnote{Note that for minimal Stinespring dilations $W$ is unitary.} Hence, $\Phi^*_1 = W^* \Phi^*_2 W$ and the claim follows.
    %criterion holds with respect to the Stinespring dilation $(\Phi_1,v_1,\cK_1)$ if and only if it holds for the Stinespring dilation $(\Phi_2,v_2,\cK_2)$.
\end{proof}
%Comparing with Thm.~\ref{thm: Schroedinger-HJW theorem}, we remark that a Stinespring dilation can always be chosen such that the ancillary state is pure. From this and Eq.~(\ref{eq: channel-state duality}) it follows that $\tilde{rho}$ can be chosen to be pure, namely in which case the result follows directly from Thm.~\ref{thm: Schroedinger-HJW theorem}.

Next, we have the following important characterisation.

\begin{lemma}\label{lm: C* and anti-C* homo}
    Let $\Phi: \cA \ra \cB(\cK)$ be a $C^*$-homomorphism. %such that also
    Then $\Phi^*: \cA \ra \cB(\cK)$ %(equivalently, $\Phi: \cA^\mathrm{op} \ra \cB(\cK)$)
    is a $C^*$-homomorphism if and only if $\Phi(\cA) \subset \cB(\cK)$ is a commutative subalgebra,\footnote{It is interesting to note that Bell's theorem holds for states over $C^*$-algebras as long as one of them is commutative \cite{Baez1987}. We discuss the relation with Bell's theorem and Bell nonlocality elsewhere \cite{DoeringFrembs2019a}.} equivalently, $\Phi = \Phi|_V$ for $V \subset \cA$ a commutative subalgebra.
\end{lemma}

\begin{proof}
    If $\Phi^*$ is a $C^*$-homomorphism, then for all $a_1,a_2 \in \cA$
    \begin{equation*}
        \Phi(a_1)\Phi(a_2) = \Phi(a_1a_2) = \Phi^*((a_1a_2)^*) = \Phi^*(a_2^*a_1^*) = \Phi^*(a_2^*)\Phi^*(a_1^*) = \Phi(a_2)\Phi(a_1)\; ,
    \end{equation*}
    hence, $\Phi(\cA) \subset \cB(\cK)$ is a commutative subalgebra. Conversely, if $\Phi(\cA) \subset \cB(\cK)$ is a commutative subalgebra, then for all $a_1^*,a_2^* \in \cA$
    \begin{equation*}
        \Phi^*(a_2^*a_1^*) =
        \Phi^*((a_1a_2)^*) = 
        \Phi(a_1a_2) =
        \Phi(a_1)\Phi(a_2) =
        \Phi(a_2)\Phi(a_1) =
        \Phi^*(a_2^*)\Phi^*(a_1^*)\; .
    \end{equation*}
    Hence, $\Phi^*$ is a $C^*$-homomorphism.
    
    Since $\Phi(\cA)\subset \cB(\cK)$ is a commutative subalgebra, there exists a maximal commutative subalgebra $V \subset \cA$ such that $\Phi(V) = \Phi(\cA)$. %\subseteq?
    We show that $\Phi(a) = 0$ for all $a \perp V$ (with respect to the Hilbert-Schmidt inner product $(a,a') := \tr[a'^* a]$). Without loss of generality, we may assume that $V$ is the commutative subalgebra generated by diagonal matrices. We want to show that $\Phi(a) = 0$ for all $a \in V^\perp$. The latter implies $\mathrm{tr}[a] = 0$, hence, $a = [b,c]$ for some $b,c \in \cA$ \cite{Shoda1936}. %Thm.~3
    We have $\Phi(a) = \Phi([b,c]) = [\Phi(b),\Phi(c)] = 0$, since $\Phi$ is a homomorphism and $\Phi(\cA)$ is commutative. Consequently, $\Phi$ acts non-trivially only on $V$ and is zero otherwise.
    %If $\Phi$ is unital, then $\Phi|_{\mc{P}(\cN)}$ is an orthomorphism; note that $\phi_\rho$ is unital (by normalisation of $\rho$)
\end{proof}
%NB: Note that this does not imply that $\cA$ is commutative - $\Phi$ only preserves commutativity, it does not necessarily reflect it!

The following key result shows that this property is equivalent to separability.
%flag: up to (minimal) Stinespring dilation

\begin{theorem}\label{thm: dilated PPT criterion}
     Let $\rho \in \SAB$, let $\phi_\rho$ be the map under the isomorphism in Eq.~(\ref{eq: CJ-isomorphism}) and let $\phi_\rho = v^*\Phi_\rho v$ be a %(minimal)
     Stinespring dilation of $\phi_\rho$. Then $\rho$ is separable if and only if $\phi^*_\rho = v^*\Phi^*_\rho v$ is a Stinespring dilation of $\phi_\rho^*$, i.e., if and only if $\Phi^*_\rho$ %\cA \ra \cB(\cK)
     is a $C^*$-homomorphism.
\end{theorem}

\begin{proof}
    By Lm.~\ref{lm: Schroedinger-HJW for maps}, we can choose the Stinespring dilation $\phi_\rho = v^* \Phi_\rho v$ to be of the following simple form: $v: \cH_\cB \ra \cK$ for $\cK = \cH_\cB \otimes \cH_\cA$\footnote{For a general Stinespring dilation, one needs $\cK = \cH_\cB \otimes \cH_\cA \otimes \cH_\cA$ (cf. \cite{Stinespring1955,Naimark1943}). However, as it turns out in the case that both $\Phi$ and $\Phi^*$ are both $C^*$-homomorphisms, $\cK$ can be reduced by one factor of $\cH_\cA$.}
    %flag: usually (by Naimark/Stinespring), $\cK = (H_A \otimes H_A) \otimes \cH_\cB$, where $(\cH_\cA \otimes \cH_\cA)$ is identified with the Hilbert space of Hilbert-Schmidt operators under the Hilbert-Schmidt norm
    and $\Phi_\rho = 1_\cB \otimes \mathrm{id}_{\mathrm{supp}(\Phi_\rho)}$.\footnote{This representation allows to interpret a quantum channel $\phi$ as a coarse-grained unitary bipartite channel on the target and some unknown ancillary system, after tracing out the latter. In particular, generalised measurements can be understood as projective measurements on a larger system \cite{Ozawa1984}.}
    In fact, if $\Phi^*_\rho$ is a $C^*$-homomorphism, then $\phi_{\rho^{T_\cA}} = \phi^*_\rho = v^*\Phi^*_\rho v$ is a Stinespring dilation (see Lm.~\ref{lm: transposition vs *}).\footnote{In particular, $\rho^{T_\cA}$ is positive in this case.} By Lm.~\ref{lm: C* and anti-C* homo}, $\Phi_\rho(\cA) \subset \cB(\cK)$ is a commutative subalgebra
    %flag: general case
    %Recall that by Gelfand duality, every unital commutative $C^*$-algebra $\cA'$ can be represented as the algebra of continuous functions $\cA' \cong C^0(X)$ on a compact Hausdorff space $X$ \cite{GelfandNaimark1943}. In what follows we will restrict to/In finite dimensions are the same as von Neumann algebras. In this case Gelfand duality asserts that every commutative von Neumann algebra $\cN$ is isomorphic to the space of bounded measurable functions $\cN \cong L^\infty(X,\mu)$ for $(X,\mu)$ a (localisable) measurable space (cf. \cite{Pavlov2020}).
    and there exists a commutative subalgebra $V \subset \cA$ such that $\Phi_\rho(V) = \Phi_\rho(\cA) \subset \cB(\cK)$. We thus have $\Phi_\rho(a) = \Phi_\rho|_V(a) = 1_\cB \otimes p_Vap_V$, where $p_V$ denotes the projection onto $V$.
    
    Next, let $(|\xi_k\rangle)_k$ be a basis of $\cH_\cA$ such that $|\xi_k\rangle\langle\xi_k| = p_k$ for all one-dimensional projections $p_k \in \mc{P}_1(V)$, and let $(|e_i\rangle)_i$ be an orthonormal basis of $\cH_\cB$. We can decompose $v: \cH_\cB \ra \cK$ by its action on basis states, $v(|e_j\rangle) := \sum_{ij} c^k_{ij}|e_i\rangle \otimes |\xi_k\rangle$ with $c^k_{ij} \in \C$, hence, $v = \sum_{ijk} c^k_{ij} |e_i\rangle\langle e_j|\otimes|\xi_k\rangle = \sum_{k} X_k \otimes |\xi_k\rangle$ with $X_k = \sum_{ij} c^k_{ij} |e_i\rangle\langle e_j|$. Consequently, for all $a = \sum_{r} a_rp_r \in V$
    \begin{align}\label{eq: decomposition as measure-prepare channel}
        \phi_\rho(a)
        %= v^* \Phi(a) v
        = v^* \Phi_\rho|_V(a) v
        = v^* (1_\cB \otimes a) v
        = \sum_{kl} X_l^*X_k\ \langle\xi_l|\sum_r a_r p_r|\xi_k\rangle
        %&= \sum_{r} X_r^*X_r \mathrm{tr}[|\xi_r\rangle\langle\xi_r| a] \langle\xi_l|a_r p_r|\xi_k\rangle \\
        &= \sum_{k} E_k\ \mathrm{tr}_{\cH_\cA}[|\xi_k\rangle\langle\xi_k|\ a]\; .
        %= \sum_{k} |\psi_k\rangle\langle\psi_k| \otimes \mathrm{tr}[|\xi_k\rangle\langle\xi_k|\ a]\; ,
    \end{align}
    Clearly, $E_k = X_k^*X_k \in \cB_+$, hence, (after normalisation $E_k \rightarrow E_k/\tr[E_k]$, $|\xi_k\rangle\langle\xi_k| \rightarrow \tr[E_k]|\xi_k\rangle\langle\xi_k|$) $\phi_\rho$ is an entanglement-breaking channel,\footnote{Entanglement-breaking channels are also called \emph{measure-prepare channels} or \emph{in Holevo form} (see \cite{HorodeckiShorRuskai2003,Holevo1998}).} equivalently $\rho$ is a separable state \cite{HorodeckiShorRuskai2003}.
    %flag: if $\phi_\rho$ is trace-preserving, we have $\mathrm{tr}[\phi_\rho(p_k)] = \mathrm{tr}[E_k] = 1$ for all $k$
    %flag: Does $\tr[\rho_\phi]=1$ imply $\phi_\rho$ is TP and/or unital?
    
    Conversely, let $\rho$ be a separable state.\footnote{Clearly, $\rho^{T_\cA}$ is positive, equivalently $\phi_{\rho^{T_\cA}}$ is completely positive in this case.}
    %NB: However, this alone does \emph{not} imply that $\phi_{T_1(\rho)} = \phi^* = v^*\Phi^* v$ is a Stinespring dilation.
    Then $\phi_\rho$ is an entanglement-breaking channel \cite{HorodeckiShorRuskai2003}, i.e., there exist states $E_k \in \mc{S}(\cB)$ and positive operators $F_k \geq 0$ such that
    \begin{equation*}
        \phi_\rho(a)
        %= \sum_{k} |\psi_k\rangle\langle\psi_k|\ \mathrm{tr}_\cA[|\xi_k\rangle\langle\xi_k|\ a]\; .
        = \sum_{k=1}^K E_k\ \mathrm{tr}_{\cH_\cA}[F_k a]\; .
    \end{equation*}
    We may extend $F_k$ to a positive operator-valued measure (POVM) $(F_k)_{k=0}^K$, $F_k \in \cA_+$ by setting $F_0 := 1 - \sum_{k=1}^K F_k$ such that $\sum_{k=0}^K F_k = 1$. %and where we define the natural embedding $\cA \hookrightarrow \cA'$, $a \mapsto 0 \oplus a$ (we will write $a$ for $0 \oplus a$).
    By Naimark's theorem \cite{Naimark1943,Stinespring1955}, $F$ admits a dilation $F = \tilde{v}^* \pi \tilde{v}$, where $\tilde{v}: \cH_{\cA} \ra \cH_{\tilde{\cA}}$ is a linear map and $(\pi_k)_{k=0}^K$, $\pi_k \in \mc{P}(\cH_{\tilde{\cA}})$ a projection-valued measure (PVM). Consequently, $\phi_\rho(a) = \tilde{\phi}_\rho(a) - E_0\ \tr[F_0 a]$, where e.g. $E_0 \propto \mathbbm{1}$ and
    \begin{align*}\label{eq: decomposition from measure-prepare channel}
        \tilde{\phi}_\rho(a)
        = \sum_{k=0}^K E_k\ \mathrm{tr}_{\cH_{\cA}}[\tilde{v}^*\pi_k\tilde{v}a]
        %&= \sum_{k} E_k\ \mathrm{tr}_\cA[\tilde{v}^*\pi(k) \tilde{\Phi}(a) \tilde{v}] \\
        = \sum_{k=0}^K E_k\ \mathrm{tr}_{\cH_{\cA}}[\tilde{v}^*\pi_k \tilde{\Phi}_\rho(a)\tilde{v}]
        %= \sum_{k=0}^K E_k\ \mathrm{tr}_{\cH_{\tilde{\cA}}}[\tilde{v}\tilde{v}^*\pi_k\tilde{\Phi}_\rho(a)\tilde{v}\tilde{v}^*]
        = \sum_{k=0}^K E_k\ \mathrm{tr}_{\cH_{\tilde{\cA}}}[\pi_k\tilde{\Phi}_\rho(a)]\; .
    \end{align*}
    %NB: use tensor product decomposition of $\tilde{v}$ in last step to cyclically move $\tr_{\cH_\cA}[...]\rangle\xi|... |\xi\rangle$ into $\tr_{\cH_\cA}[...] \tr[|\xi\rangle\langle\xi|...] = \tr_{\cH_{\tilde{\cA}}}[...]$
    Here, $\tilde{\Phi}_\rho: \cA \ra \tilde{\cA} = \cB(\cH_{\tilde{\cA}})$, $\tilde{\Phi}_\rho(a) := \tilde{v}a\tilde{v}^*$ is the natural embedding under the isometry $\tilde{v}$ (that is, $\tilde{v}^*\tilde{v} = 1$), and we used that $\tilde{\Phi}_\rho \tilde{v}\tilde{v}^* = \tilde{\Phi}_\rho$,
    %flag: in particular, $\tilde{\Phi}_\rho$ spans an algebra isomorphic to $\cA$ in $\tilde{\cA}$
    where
    $\tilde{v}\tilde{v}^* \in \mc{P}(\cH_{\tilde{\cA}})$, $(\tilde{v}\tilde{v}^*)\cH_{\tilde{\cA}} \cong \cH_\cA$ in the last step. %such that %$\tilde{\Phi}(\cA) = \tilde{\Phi}(\tilde{v}\tilde{v}^* \cA \tilde{v}\tilde{v}^*)$ as well as
    %$\tilde{v}\tilde{v}^* \pi \tilde{v}\tilde{v}^* = \pi$.}
    %In particular, the second last line follows by trivial extension of $\pi$ and $\tilde{\Phi}$ to all of $\cA'$.
    %NB: Note also that $\tilde{v}^*\tilde{v} = \phi(1) = \tilde{v}^*\tilde{\Phi}(1)\tilde{v}$. Clearly, $\tilde{\Phi}$ is a $C^*$-homomorphism.
    Note that $(a,a') := \mathrm{tr}_{\cH_{\tilde{\cA}}}[a'^*a]$ defines an inner product on $\tilde{\cA}$. It follows that we can restrict the action of $\tilde{\Phi}_\rho$ to the pre-image of the commutative subalgebra $W := \langle\pi_k\rangle_{k=0}^K \subset \tilde{\cA}$, spanned by the projections $\pi_k$. More precisely, we define $\Phi_\rho: \cA \ra \cB(\mc{K})$ for $\mc{K} = \cH_\cB \otimes \cH_{\tilde{\cA}}$ by
    \begin{equation*}
        \Phi_\rho(a) =
        \begin{cases}
            1_\cB \otimes \tilde{\Phi}_\rho &\mathrm{for\ all\ } a \in \tilde{\Phi}^{-1}_\rho(W)\\
            0 &\mathrm{otherwise}
        \end{cases}\; .
    \end{equation*}
    Clearly, $\Phi_\rho$ is a $C^*$-homomorphism since $\tilde{\Phi}_\rho|_{\tilde{\Phi}^{-1}_\rho(W)}$ is. Moreover, $\Phi_\rho(\cA) \subset \cB(\cK)$ is a commutative subalgebra by construction,
    %NB: this implies that $\cA'$ can be chosen to be of minimal dimension, corresponding to the dimension of the maximal commutative subalgebra $V \subset \cA$
    hence, $\Phi^*_\rho$ is a $C^*$-homomorphism by Lm.~\ref{lm: C* and anti-C* homo}.
    %It follows that $\pi_k \tilde{\Phi}_\rho|_{\tilde{\Phi}^{-1}_\rho(W)} = \tilde{\Phi}_\rho|_{\tilde{\Phi}^{-1}_\rho(W)} \pi_k$ such that
    %\begin{equation*}
        %\tilde{\phi}_\rho
        %= \sum_{k=0}^K E_k\ \mathrm{tr}_{\cH_{\cA}}[\tilde{v}^*\pi_k \tilde{\Phi}_\rho|_{\tilde{\Phi}^{-1}_\rho(W)}\pi_k\tilde{v}]\; .
    %\end{equation*}}
    Finally, we obtain Stinespring dilations $\phi_\rho = v^*\Phi_\rho v$ and $\phi^*_\rho = v^*\Phi^*_\rho v$ as before: define $v = \sum_{k=1}^K X_k \otimes |\xi_k\rangle$ with $X_k^*X_k = E_k$ and with $(|\xi_k\rangle)_{k=0}^K$ the orthonormal basis of $\cH_{\tilde{\cA}}$, corresponding to the commutative subalgebra $W \subset \tilde{\cA}$, i.e., $|\xi_k\rangle\langle\xi_k|=\pi_k$.
    %flag: old version/general case
    %More generally, it factorises/is of the form $\phi = \rho \circ \varphi$ for $\varphi: \cA \ra X$ a (generalised) measurement (vN algebra homomorphism into a commutative algebra $L^\infty(X,\mu)$ for $(X,\mu)$ a (localisable) measurable space?) and $\rho: X \ra \cB_+$ a positive operator-valued measure}. [By Naimark's theorem, there exist a (bounded) linear map $v: $ and a projection-valued measure $\pi: X \ra \cB(\cK)$ such that $\rho = v^* \pi v$. Consequently,
    %\begin{align*}
        %\phi = \rho \circ \varphi = (v^* \pi v) \circ \varphi = v^* (\pi \circ \varphi) v = v^* \Phi v\; ,
    %\end{align*}
    %from which we deduce that $\Phi = \pi \circ \varphi$. But then $\Phi^* = \pi^* \circ \varphi = \pi \circ \varphi^*$ is also a $C^*$-algebra homomorphism.]
\end{proof}

We have used the fact that a state is separable if and only if its corresponding quantum channel in Eq.~(\ref{eq: CJ-isomorphism}) is a measure-prepare channel \cite{HorodeckiShorRuskai2003}. The latter requires a decomposition of $v$ as in Eq.~(\ref{eq: decomposition as measure-prepare channel}). %and Eq.~(\ref{eq: decomposition from measure-prepare channel})
Clearly, such a decomposition exists for any commutative subalgebra $V \subset \cA$. In turn, Thm.~\ref{thm: dilated PPT criterion} shows that such a decomposition exists for all of $\cA$ if and only if $\Phi_\rho(\cA)$ is a commutative subalgebra.
%flag: relate to Wigner-Weyl transform?!

Moreover, Thm.~\ref{thm: dilated PPT criterion} sheds new light on the reason why the PPT criterion is not necessary in general: let $\phi_\rho = v^*\Phi_\rho v$ be a Stinespring dilation and let $\phi^*_\rho = \phi_{\rho^{T_\cA}}$ be completely positive, this does not imply that $\phi^*_\rho = v^*\Phi^*_\rho v$ is a Stinespring dilation for $\phi^*_\rho$.%\footnote{In turn, Thm.~\ref{thm: dilated PPT criterion} explains why the PPT criterion is also necessary whenever $\phi_\rho$ is decomposable \cite{Horodeckisz1996}; note that this is the case for every positive map $\phi: \cA \ra \cB$ whenever $\mathrm{dim}(\cH_\cA) = 2$ and $\mathrm{dim}(\cH_\cB) = 2,3$ \cite{Stormer1963,Woronowicz1976}.}

We record the following corollary of Thm.~\ref{thm: dilated PPT criterion}.

\begin{corollary}\label{cor: commutative extension problem}
    Let $\rho \in \SAB$, let $\phi_\rho$ the map under the isomorphism in Eq.~(\ref{eq: CJ-isomorphism}), and let $\phi_\rho = v^*\Phi_\rho v$ be a %(minimal)
    Stinespring dilation. Then $\rho$ is separable %if and only if $\phi_\rho(\cA) \subset \cB$ can be extended to a commutative algebra; 
    if and only if $\phi_\rho = \phi_\rho|_V$ for a commutative subalgebra $V \subset \cA$.
\end{corollary}
%flag: compare with Prop.~\ref{prop: dilated PPT criterion - bipartite, state case}

\begin{proof}
    This follows immediately from Lm.~\ref{lm: Schroedinger-HJW for maps}, Thm.~\ref{thm: dilated PPT criterion}, and Lm.~\ref{lm: C* and anti-C* homo}.
\end{proof}

We surmise that Thm.~\ref{thm: dilated PPT criterion}---especially in the form of Cor.~\ref{cor: commutative extension problem}---entails improvements of existing protocols for practical verification of entanglement, e.g. in the form of semi-definite linear programmes in \cite{DohertyParriloSpedalieri2002,DohertyParriloSpedalierei2004}. We leave this as an exciting direction for future research. In the remainder, we focus on the physical content of Thm.~\ref{thm: dilated PPT criterion} in terms of the arrow of time.

%% file: TheArrowOfTime/JordanAlgebras.tex
\textbf{Jordan algebras.} We recall that an abstract Jordan algebra $\cJ$ is an algebra over a field with a product that satisfies $a \circ b = b \circ a$ and $(ab)(aa) = (a(b(aa))$ for all $a,b \in \cJ$.\footnote{For an extensive study of Jordan algebras, see \cite{McCrimmon2004}.}. Given an associative algebra $\cA$ one obtains a Jordan algebra $\cJ(\cA)$ by symmetrisation. %$a \circ b := \frac{1}{2}(ab + ba)$.
If $\cJ = \cJ(\cA)$ for an associative algebra $\cA$, then the Jordan algebra is called \emph{special}, otherwise it is called \emph{exceptional}.\footnote{The prototypical exceptional Jordan algebra is the so-called Albert algebra $H_3(\mathbb{O})$ \cite{Albert1934}.}
%Note that $\cJ$ is a real algebra. We will often consider its complexification and denote is similarly by $\cJ$, the distinction will be clear from context.
In particular, every $C^*$- (and von Neumann)\footnote{Recall that a $C^*$-algebra is an involutive Banach algebra (closed in norm) satisfying the defining $C^*$-property, $||x^*x|| = ||x||^2$. %(In fact, by the $C^*$-property the norm is determined algebraically from $S$.)
A von Neumann algebra is a $C^*$-algebra closed in the weak operator topology.} algebra defines a JB(W) algebra: a JB(W) algebra is a (weakly closed) Jordan algebra that is also a Banach space ($||a \circ b|| \leq ||a||\cdot||b||$) such that $||a^2|| = ||a||^2 \leq ||a^2 + b^2||$. For simplicity, here we only consider matrix algebras over the complex numbers, $\cA = M_n(\C)$, $n \in \mathbb{N}$.
%flag[- for general case, see \input{EntanglementAndtheArrowOfTime/CCC}]
In this case, the set of Hermitian matrices $H_n(\C)$ under the anti-commutator $\{a,b\} := ab + ba$ defines a real Jordan algebra $\cJ(\cA)_\mathrm{sa} := (H_n(\C),\{\cdot,\cdot\})$. We denote its complexification by $\cJ(\cA) = \cJ(\cA)_\mathrm{sa} +i\cJ(\cA)_\mathrm{sa}$.

Crucially, Jordan products are commutative. As such the Jordan algebra $\cJ(\cA)_\mathrm{sa}$ is the same as the Jordan algebra $\cJ(\cA^\op)_\mathrm{sa}$ of the opposite algebra $\cA^{\op}$, i.e., the algebra obtained from $\cA$ by reversing the order of composition (matrix multiplication), %and taking Hermitian adjoints}
%NB: This corresponds to the dual in the category-theoretic sense.
\begin{align}\label{eq: two associative products}
    \begin{split}
        \cA &:= \{a \in \cA \mid \forall a_1,a_2 \in \cA:\ a_1 \cdot_+ a_2 = \frac{1}{2}\{a_1,a_2\} + \frac{1}{2}[a_1,a_2]\}\; ,\\
        \cA^\mathrm{op} &:= \{a \in \cA \mid \forall a_1,a_2 \in \cA:\ a_1 \cdot_- a_2 = \frac{1}{2}\{a_1,a_2\} - \frac{1}{2}[a_1,a_2]\}\; .
    \end{split}
\end{align}
%In other words, $\cJ(\cA)_\mathrm{sa} \cong \cJ(\cA^\mathrm{op})_\mathrm{sa}$ is the real, symmetric part of both,
The difference between the associative algebras $\cA$ and $\cA^\mathrm{op}$ %$\cA_- \cong (\cA^\op)^*$}
is the anti-symmetric part or commutator. In order to extract from this a notion of time directionality, we relate commutators to (infinitesimal) symmetries of $\cJ(\cA)_\mathrm{sa}$.\\

%Note first that we can consider both the complexification and opposite algebra of a Jordan algebra. If the Jordan algebra is the symmetric part of an associative algebra, these two structures are related: for every Jordan algebra $\mc{J}$ there exists a unique homomorphism into a von Neumann algebra such that $\mc{J}$ is the fixed point of an involution of order two (cf. \cite{Hanche-OlsenStormer1984JordanOA}). (In particular, for Jordan matrix algebras this involution is the $*$-operation, corresponding to complex conjugation and transposition)
%Furthermore, while the above result leaves some room for exceptional Jordan algebras, the gap is narrow and corresponds precisely to the existence of an orientation for a Jordan algebra.
%Nevertheless, every Jordan $*$-homomorphism between JW algebras already preserves commutators up to a choice of sign. More precisely, note that both $\cN$ and its opposite algebra $\cN^o$ reduce to $\JN$ as JW algebras, and these are the only associative algebras with this property (cf. \cite{Kadison1951,AlfsenShultz1998}).\footnote{In general, the resulting algebras are not isomorphic as von Neumann algebras \cite{Connes1975}.} The respective associative products are distinguished by a sign, $\frac{1}{2}(ab + ba) \pm \frac{1}{2}(ab - ba)$. In order to identify $\cN$ from $\JN$ we thus have to specify this sign.

%% file: TheArrowOfTime/Orientations.tex
\textbf{Time orientations.} Dynamics is naturally expressed in terms of one-parameter groups of Jordan automorphisms $\R \ni t \mt \mathrm{Aut(\cJ(\cA)_\mathrm{sa})}$. Recall that for $\cA = \cB(\cH_A)$ (in particular, for $\cA = M_n(\C)$) every such one-parameter group is given by conjugation with a unitary or anti-unitary operator by Wigner's theorem \cite{Wigner1931_GruppentheorieUndQM,Bargmann1964}. In fact, Wigner's theorem holds on the level of Jordan algebras \cite{LandsmanLindenhovius2018,DoeringFrembs2019a}. In $\cA$, we obtain one-parameter groups of the form
\begin{equation}\label{eq: one-parameter groups}
    %(e^{t\psi(a_1)}a_2 :=) 
    e^{t\ad(ia_1)}(a_2)
    = e^{ita_1}a_2e^{-ita_1} \quad \quad \forall t \in \mathbb{R}, a_1,a_2 \in \cJ(\cA)_\mathrm{sa}\; .
\end{equation}
If we interpret $a_1$ as the Hamiltonian of the system, then Eq.~(\ref{eq: one-parameter groups}) is just the standard expression for unitary evolution, in which $t$ plays the role of a \emph{time parameter}. More generally, Eq.~(\ref{eq: one-parameter groups}) defines a one-parameter group for every element $a \in \cJ(\cA)_\mathrm{sa}$. In particular, note that for every $a \in \cJ(\cA)_\mathrm{sa}$ and $\lambda \in \R_+$ also $\lambda a \in \cJ(\cA)_\mathrm{sa}$. Hence, we cannot give physical meaning to the absolute value of $t$ without first specifying a Hamiltonian $a$. 

Nevertheless, the sign of $t$ carries physical meaning independent of the choice of $a \in \cJ(\cA)_\mathrm{sa}$. To see this, we remark that inherent in Eq.~(\ref{eq: one-parameter groups}) is the canonical identification between self-adjoint operators (observables) and generators of Jordan automorphisms (symmetry generators), %in terms of multiplication by the complex unit
$a \mapsto \ad(ia)$ for all $a \in \cJ(\cA)_\mathrm{sa}$ \cite{GrginPetersen1974,AlfsenShultz1998,Baez2020}. Moreover, note that changing this identification to $a \mapsto \ad(-ia)$ results in a sign change for the parameter $t$ in Eq.~(\ref{eq: one-parameter groups}), equivalently to a change in the commutator %$\ad(ia)(a') \ra \ad(-ia)(a')$
and thus to a change in the order of composition from $\cA$ to $\cA^\mathrm{op}$ in Eq.~(\ref{eq: two associative products}) (see also Lm.~\ref{lm: Jordan to C*-homo} below). 

In contrast, in $\cJ(\cA)$ there is no canonical identification between self-adjoint operators and generators of Jordan automorphisms \cite{AlfsenShultz1998,Hanche-OlsenStormer1984JordanOA}. %Consequently, the fact that every unitary operator $u \in \cA$ is generated by a self-adjoint operator $u = e^{ia}$ for some $a \in \mc{J}(\cA)_\mathrm{sa}$ requires additional input.
Consequently, in $\cJ(\cA)$ we cannot interpret the sign of the parameter $t$ in the corresponding one-parameter groups independently of the choice of Hamiltonian $a \in \cJ(\cA)_\mathrm{sa}$. By comparison, lifting $\cJ(\cA)$ to $\cA$ thus equips the latter with an intrinsic direction of time, mediated by the identification $a \mapsto \ad(ia)$.\footnote{Generalising the mapping $a \mapsto \ad(ia)$, \cite{AlfsenShultz1998a} characterise those maps which lift JB(W) to $C^*$ (von Neumann) algebras. By their physical interpretation, such maps are called \emph{dynamical correspondences}.} To emphasise this distinction, we define the \emph{canonical time orientation $\Psi_\cA$} on $\cJ(\cA)$ by\footnote{The notion of time orientation was introduced in \cite{Doering2014,FrembsDoering2022a}.}
\begin{equation}\label{eq: canonical time orientation}
    \Psi_\cA := \Ad: \R \times \cJ(\cA)_\mathrm{sa} \ni (t,a) \mapsto e^{t\ad(ia)},
\end{equation}
and call $\cA_+ := (\cJ(\cA),\Psi_\cA)$ the observable Jordan algebra together with its canonical time orientation. Similarly, we define the \emph{reverse time orientation} by %(cf. Prop.~15 in \cite{AlfsenShultz1998a})
\begin{equation}\label{eq: reverse time orientation}
    \Psi^*_\cA := * \circ \Psi_\cA: \R \times \cJ(\cA)_\mathrm{sa} \ni (t,a) \mapsto e^{-t\ad(ia)}\; ,
\end{equation}
and set $\cA_- := (\cJ(\cA),\Psi^*_\cA)$.\footnote{By Thm.~23 in \cite{AlfsenShultz1998a}, these are the only time orientations on $\cJ(\cA)$, deriving from $\cA$ and $\cA^\mathrm{op}$, respectively.} %By Thm.~23 in \cite{AlfsenShultz1998a}, $\cA_+ \cong \cA$ and $\cA_- \cong \cA^\mathrm{op}$.\\ %more precisely, $\cA^\mathrm{op} \cong (\cJ(\cA)^*,\Psi^*_\cA)$, hence, the isomorphism is the complex conjugation on $\cJ(\cA)$

%On $\cJ(\cA)_\mathrm{sa}$, one can define a Lie algebra of so-called \emph{skew order derivations} $\mathrm{OD}_s(\cJ(\cA)_\mathrm{sa})$ \cite{Connes1974,AlfsenShultz1998}. In $\cA$, this Lie algebra simply becomes $\mathrm{OD}_s(\cJ(\cA)_\mathrm{sa}) = i\cJ(\cA)_\mathrm{sa}$ with Lie product $\ad: i\cJ(\cA)_\mathrm{sa}\times i\cJ(\cA)_\mathrm{sa} \ra i\cJ(\cA)_\mathrm{sa}$ given by $\ad(a)(b) = [a,b]$.

%\begin{definition}\label{def: orientation}
    %Let $\cJ(\cA)$ be the Jordan (JB(W)) algebra corresponding to a (unital $C^*$- (von Neumann)) algebra $\cA$. Then $\cA$ defines a \emph{canonical time orientation $\Psi_\cA$} on $\cJ(\cA)$ by
    %\begin{align*}
        %\Psi: \R \times \cJ(\cA) &\lra \mathrm{Aut(\cJ(\cA))} \\
        %(t,a) &\lmt \Ad(ita) = e^{\ad(ita)}\; .
    %\end{align*}
    %In particular, a (unital $C^*$- (von Neumann)) algebra is a Jordan (JB(W)) algebra together with its canonical time orientation, $\cA = (\cJ(\cA),\Psi_\cA)$.\footnote{The notion of time orientation was introduced in \cite{Doering2014}.}
%\end{definition}
%\input{GeneralCase}

%% file: TheArrowOfTime/EntanglementAndTimeReversal.tex
\textbf{Entanglement and time orientation.} Returning to Thm.~\ref{thm: dilated PPT criterion}, we are interested in the difference between Jordan $*$-homomorphism and $C^*$-homomorphism. Recall that a \emph{Jordan $*$-homomorphism} $\Phi: \cJ(\cA) \ra \cJ(\cB)$ is a linear map preserving the Hermitian adjoint, $* \circ \Phi = \Phi \circ *$, equivalently, $\Phi|_{\cJ(\cA)_\mathrm{sa}}: \cJ(\cA)_\mathrm{sa} \ra \cJ(\cB)_\mathrm{sa}$, and the Jordan product, i.e., $\Phi(\{a_1,a_2\}) = \{\Phi(a_1),\Phi(a_2)\}$ for all $a_1,a_2 \in \mc{J}(\cA)_\mathrm{sa}$.
Consequently, $\Phi: \cJ(\cA) \ra \cJ(\cB)$ lifts to a $C^*$-homomorphism $\Phi: \cA \ra \cB$ if and only if it preserves commutators, $\Phi([a_1,a_2]) = [\Phi(a_1),\Phi(a_2)]$. Using Eq.~(\ref{eq: canonical time orientation}), we re-express this condition in terms of one-parameter groups of Jordan automorphisms.

\begin{lemma}\label{lm: Jordan to C*-homo}
    Let $\Phi: \cJ(\cA) \ra \cJ(\cB)$ be a Jordan $*$-homomorphism. Then $\Phi: \cA \ra \cB$ lifts to a $C^*$-homomorphism if and only if it preserves the canonical time orientations $\Psi_\cA$ and $\Psi_\cB$,
    \begin{equation}\label{eq: C*-homo preserves orientations}
        \forall t \in \R, a \in \cJ(\cA)_\mathrm{sa}:\ \Phi \circ \Psi_\cA(t,a) %= \Phi \circ e^{t\ad(ia)} = e^{t\ad(i\Phi(a))} \circ \Phi
        = \Psi_\cB(t,\Phi(a)) \circ \Phi\; .
    \end{equation}
\end{lemma}

\begin{proof}
    Clearly, a $C^*$-homomorphism $\Phi$ preserves Eq.~(\ref{eq: C*-homo preserves orientations}). Conversely, by differentiation, %(cf. \cite{DeSeguinsPazzis2010})
    \begin{align*}
        &\left.\frac{d}{dt}\right\vert_{t=0} \left(\Phi \circ \Psi_\cA(t,a_1)\right)(a_2)
        = \left.\frac{d}{dt}\right\vert_{t=0} \left(\Psi_\cB(t,\Phi(a_1)) \circ \Phi\right)(a_2) \\
        \Leftrightarrow \quad & \Phi\left(\left.\frac{d}{dt}\right\vert_{t=0} e^{ita_1}a_2e^{-ita_1}\right) = \left.\frac{d}{dt}\right\vert_{t=0} e^{it\Phi(a_1)}\Phi(a_2)e^{-it\Phi(a_1)}\\
        \Leftrightarrow \quad &\Phi([a_1,a_2]) = [\Phi(a_1),\Phi(a_2)]\; .
    \end{align*}
    for all $a_1,a_2 \in \cJ(\cA)_\mathrm{sa}$. %(and thus for all $a_1,a_2 \in \mc{J}(\cA)$ by complexification)
    $\Phi$ thus preserves commutators, hence, is a $C^*$-homomorphism.
\end{proof}
%flag: More generally, it is easy to see that orientations correspond to a binary choice of time direction in every factor of a von Neumann algebra \cite{?}
%\begin{lemma}\label{lm: S classifies time orientations}
    %Classify all orientations in terms of $S$. 
%\end{lemma}

%NB: On the other hand, a linear map $\Phi: \cA \ra \cB$ is a $C^*$-homomorphism if and only if it commutes with $S_\cA, S_\cB$, i.e., $\forall a,b \in \cA$
%\begin{equation*}
    %(\Phi \circ S_\cA)(ab)
    %= \Phi(b^*a^*) = \Phi(b^*)\Phi(a^*) = \Phi^*(b)\Phi^*(a) = S_\cB(\Phi(a)\Phi(b)) = (S_\cB \circ \Phi)(ab)\; .
%\end{equation*}
%Hence, we conclude that specifying a complex structure $S$ on $\cJ$ is equivalent to a dynamical correspondence/orientation.}\\

Assume $\Phi: \cA \ra \cB(\cK)$ in Eq.~(\ref{eq: C*-homo preserves orientations}) is part of a Stinespring dilation $\phi_\rho = v^*\Phi v$ for the image of a bipartite state $\rho \in \mc{S}(\cA \otimes \cB)$ under the isomorphism in Eq.~(\ref{eq: CJ-isomorphism}). Since $\cB$ arises from $\cB(\cK)$ by restriction under $v$, the time orientation $\Psi_\cB$ on $\cB$ uniquely lifts to a time-orientation $\Psi'_\cB$ on $\cB(\cK)$.

This motivates the following definition, which first appeared in the context of classifying quantum states from non-signalling bipartite correlations \cite{FrembsDoering2022a} (see also Sec.~\ref{sec: quantum from non-signalling correlations}).

\begin{definition}\label{def: time-oriented state}
    Let $\rho \in \SAB$. $\rho$ is called \emph{time-oriented with respect to $\cA_- \cong (\cJ(\cA),\Psi^*_\cA)$} and $\cB_+ = (\cJ(\cB),\Psi_\cB)$ if and only if $\Phi_\rho: \cA \ra \cB(\cK)$ in $\phi_\rho = v^*\Phi_\rho v$ preserves time orientations $\Psi^*_\cA = * \circ \Psi_\cA$ and $\Psi'_\cB$,
    \begin{equation*}
        \forall t \in \R, a \in \cJ(\cA)_\mathrm{sa}:\ \Phi_\rho \circ \Psi^*_\cA(t,a) = \Psi'_\cB(t,\Phi(a)) \circ \Phi_\rho\; .
    \end{equation*}
\end{definition}

We remark that the appearance of the reverse time orientation $\Psi^*_\cA$ in Def.~\ref{def: time-oriented state} is a consequence of the identification of bipartite quantum states and quantum channels via Choi's theorem \cite{Choi1975} (for more details, see \cite{Frembs2022b}). Def.~\ref{def: time-oriented state} is the missing piece of physical data to identify bipartite non-signalling distributions with quantum states \cite{FrembsDoering2022a}. Together, this shows that quantum states encode information about the relative time orientation between subsystems.

This is a genuine quantum effect. What is more, it is intimately related with entanglement: in fact, Def.~\ref{def: time-oriented state} allows us to reformulate the separability criterion in Thm.~\ref{thm: dilated PPT criterion} in terms of time orientations.

\begin{theorem}\label{thm: entanglement and time}
    A bipartite state $\rho \in \cS(\cA \otimes \cB)$ is separable if and only if it is time-oriented with respect to $\cA_- = (\cJ(\cA),\Psi^*_\cA)$ and $\cB_+ = (\cJ(\cB),\Psi_\cB)$ as well as $\cA_-$ and $\cB_- = (\cJ(\cB),\Psi^*_\cB)$.
\end{theorem}

\begin{proof}
    By Thm.~\ref{thm: dilated PPT criterion}, $\rho$ is separable if and only if $\Phi_\rho$ and $\Phi^*_\rho$ are $C^*$-homomorphisms for any Stinespring dilation $\phi_\rho = v^*\Phi_\rho v$. Since $C^*$-homomorphisms preserve time orientations by Lm.~\ref{lm: Jordan to C*-homo}, $\rho$ is time-oriented with respect to both $\cA_-$ and $\cB_+$ as well as $\cA_-$ and $\cB_-$.\footnote{Note that since $\Phi_\rho: \cA \ra \cB(\cK)$ is Hermiticity-preserving, i.e., $\Phi^*_\rho(a) = \Phi_\rho(a^*)$ for all $a \in \cA$, $\Phi_\rho$ preserves time orientations $\Psi_\cA$ and $\Psi_\cB$ if and only if it preserves time orientations $\Psi^*_\cA$ and $\Psi^*_\cB$; similarly $\Phi_\rho$ preserves time orientations $\Psi^*_\cA$ and $\Psi_\cB$ if and only if it preserves time orientations $\Psi_\cA$ and $\Psi^*_\cB$.}
    %First, note that according to Def.~\ref{def: time-oriented state}, every bipartite state is time-oriented with respect to $\cA$ and $\cB$. Next, recall that we found a necessary and sufficient criterion for separability in Thm.~\ref{thm: dilated PPT criterion}: $\rho$ is separable if and only if $\Phi^*$ is also a $C^*$-homomorphism. By definition, $\rho$ is thus time-oriented also with respect to $\cA$ and $\cB^*$.
    
    Conversely, $\rho$ is time-oriented with respect to $\cA_-$ and $\cB_+$ by Thm.~3 in \cite{FrembsDoering2022a}, i.e.,
    \begin{equation}\label{eq: orientation 1}
        \forall t \in \R, a \in \cJ(\cA)_\mathrm{sa}:\ \Phi_\rho \circ \Psi^*_\cA(t,a)
        = \Phi_\rho \circ \Psi_\cA(-t,a)
        = \Psi'_\cB(t,\Phi(a)) \circ \Phi_\rho\; .
    \end{equation}
    %hence, $\Phi([a_1,a_2]) = - [\Phi(a_1),\Phi(a_2)]$ by Lm.~\ref{lm: Jordan to C*-homo}.
    where $\Psi^*_\cA(t,a) = * \circ \Psi_\cA(t,a) = \Psi_\cA(-t,a)$ by Eq.~(\ref{eq: reverse time orientation}). In particular, $\Phi_\rho$ in $\phi_\rho = v^*\Phi_\rho v$ is a $C^*$-homomorphism \cite{Stinespring1955}. If $\rho$ is also time-oriented with respect to $\cA_-$ and $\cB_-$, then by Def.~\ref{def: time-oriented state},
    \begin{align}\label{eq: orientation 2}
        \begin{split}
        &\forall t \in \R, a \in \cJ(\cA)_\mathrm{sa}:\ \Phi_\rho \circ \Psi^*_\cA(t,a) = \Psi'^*_\cB(t,\Phi(a)) \circ \Phi_\rho \\
        \Longleftrightarrow \quad 
        &\forall t \in \R, a \in \cJ(\cA)_\mathrm{sa}:\ \Phi_\rho \circ \Psi_\cA(t,a) = \Psi'_\cB(t,\Phi(a)) \circ \Phi_\rho\; .
        \end{split}
    \end{align}
    %flag: write explicitly in terms of exponentials ...
    %\begin{align*}
        %%&\forall t \in \R, a,b \in \cJ(\cA):\ (\Phi \circ \Psi_\cA(t,a))(b) = (\Psi_\cB(-t,\Phi(a)) \circ \Phi)(b)}\\ \Longleftrightarrow\
        %&\forall t \in \R, a_1,a_2 \in \cJ(\cA)_\mathrm{sa}:\ (\Phi \circ \Psi_\cA(t,a_1))(a_2) = (\Psi_\cB(-t,\Phi(a_1)) \circ \Phi)(a_2)\\
        %\Longleftrightarrow\ &\forall t \in \R, a_1,a_2 \in \cJ(\cA)_\mathrm{sa}:\ \Phi(e^{ita_1}a_2e^{-ita_1}) = e^{-it\Phi(a_1)}\Phi(a_2)e^{it\Phi(a_1)}\\
        %\Longleftrightarrow\ &\forall t \in \R, a_1,a_2 \in \cJ(\cA)_\mathrm{sa}:\ \Phi(e^{ita_1})\Phi(a_2)\Phi(e^{-ita_1}) = \Phi(e^{-ita_1})\Phi(a_2)\Phi(e^{ita_1})\; .
    %\end{align*}
    Differentiating Eq.~(\ref{eq: orientation 1}) and Eq.~(\ref{eq: orientation 2}) yields %\Phi_\rho([a_1,a_2]) = -
    $[\Phi_\rho(a_1),\Phi_\rho(a_2)] = -[\Phi_\rho(a_1),\Phi_\rho(a_2)] = 0$ for all $a_1,a_2 \in \mc{J}(\cA)_\mathrm{sa}$ (cf. Lm.~\ref{lm: Jordan to C*-homo}). It follows that $\Phi_\rho(\cA) \subset \cB$ is a commutative subalgebra, by Lm.~\ref{lm: C* and anti-C* homo} $\Phi^*_\rho$ is therefore a $C^*$-algebra homomorphism and by Thm.~\ref{thm: dilated PPT criterion} $\rho$ is separable.
\end{proof}

%% file: TheArrowOfTime/TTTandCCC.tex
In this section, we embed the classification of bipartite entanglement in terms of compatibility with time orientations in local observable algebras (Thm.~\ref{thm: entanglement and time}) into a wider context. We especially focus on time orientations as a complex structure on $\cJ(\cA)$, as well as their role within the intrinsic, thermodynamic arrow of time in von Neumann algebras.\\

%flag: see also \input{Addita/ComplexStructure}
\textbf{Time orientations and complex structure.} Following \cite{Doering2014,FrembsDoering2022a}, we have expressed the difference between $\cJ(\cA)$ and $\cA$ in terms of time orientations in Def.~\ref{def: time-oriented state}. As exponentials of dynamical correspondences \cite{AlfsenShultz1998a}, time orientations highlight the double role played by self-adjoint operators: as observables and generators of dynamics in quantum mechanics \cite{GrginPetersen1974,AlfsenShultz1998}. This perspective has some appeal when considering axiomatic reconstructions of quantum mechanics and possible generalisations they suggest.

For instance, note that a `quantum formalism' can be defined also over the real instead of the complex numbers (see e.g. \cite{Stueckelberg1960,HardyWootters2011}). More generally, several results aiming to reconstruct quantum mechanics arrive at the level of (special) Jordan algebras corresponding to associative algebras over the real, complex and quaterionic numbers (e.g. \cite{JordanVNWigner1934,GrginPetersen1974,Koecher1999,Soler1995}). In this context, it is interesting to note that time orientations define a complex structure on (the order derivations of) $\cJ(\cA)$ \cite{Connes1974,AlfsenShultz1998a}. Compare this with Eq.~(\ref{eq: canonical time orientation}), where we used the complex structure of the associative algebra $\cA$ implicitly to define the canonical time orientation $\Psi_\cA$. %Note that conjugation $S$ commutes with dynamical correspondences/orientations, $S \circ \Psi = \Psi \circ S$: for all $a,b \in \cJ(\cA)$, $t \in \R$:
%\begin{equation*}
    %(S \circ \Psi(t,a))(b)
    %= S(e^{ita}be^{-ita})
    %= e^{ita}S(b)e^{-ita}
    %= (\Psi(t,a) \circ S)(b)\; .
%\end{equation*}
In this way, dynamical correspondences can be seen as a justification for the prominence of complex numbers in quantum mechanics \cite{Baez2012}.\footnote{For the intimate relationship between dynamical correspondences and Noether's theorem, see \cite{Baez2020}.} By Thm.~\ref{thm: entanglement and time}, these arguments are further inherently connected with quantum entanglement.\\
%flag: compare this with [NavascuesEtAl2021]?!

\textbf{Outlook: intrinsic dynamics and thermal time.} One of the deepest insights into the emergence of time from purely algebraic considerations arises in infinite dimensions and the structure theory of (hyperfinite) von Neumann algebras \cite{Connes1975,Connes1985,Haagerup1987,Haagerup2016}. The latter heavily rests on the foundational insights by Tomita and Takesaki \cite{Takesaki1970,TakesakiII}.

Given a von Neumann algebra $\cN$ and a faithful normal state $\omega \in \SN$, $\omega$ becomes a cyclic and separating vector $\Omega$ in its Gelfand-Naimark-Segal (GNS) representation. The operator defined by $S_\omega a\Omega := a^*\Omega$ for all $a \in \cN$ is closable, hence, has a polar decomposition $S_\omega = J_\omega \Delta^{\frac{1}{2}}_\omega$, where $J_\omega$ is an anti-unitary involution and $\Delta_\omega$ is a self-adjoint, positive operator. The fundamental results of \emph{Tomita-Takesaki theory} are summarised in the statements $J_\omega \cN J_\omega = \cN'$, where $\cN'$ is the commutant of $\cN$, and $\Delta_\omega^{it}\cN \Delta_\omega^{-it} = \cN$ for all $t \in \mathbb{R}$ \cite{Takesaki1970}. The latter implies that every faithful normal state $\omega \in \SN$ defines a one-parameter group of automorphisms $\sigma^\omega: \R \ra \mathrm{Aut}(\cN)$, $\sigma^\omega_t(a) \mt \Delta^{it} a \Delta^{-it}$, called the modular automorphism group of $\omega$.

Crucially, $S_\omega$ and thus $\sigma^\omega$ are state-dependent since they are defined with respect to the support of the state $\omega$. Despite this fact, the difference between $\sigma^\omega_t$ and $\sigma^{\omega'}_t$ is merely an inner automorphism $\sigma^\omega_t(a) = u_t\sigma^{\omega'}_t(a)u^{-1}_t$ for all $a \in \cN$, where the unitaries $(u_t)_{t \in \R}$ satisfy \emph{Connes' cocycle condition} $\sigma^\omega_{s+t} = u_s\sigma^\omega_t(u_t)$ \cite{Connes1973}. As a consequence, $\cN$ carries an \emph{intrinsic, i.e., state-independent notion of dynamics}, given by (the subgroup of) the automorphism group generated by the $\sigma^\omega$. In contrast, one can also study the operators $S_\omega$ in JBW algebras. However, without the existence of a dynamical correspondence (equivalently, time orientation), a JBW algebra cannot distinguish between the one-parameter families of automorphisms $\sigma^\omega_t$ and $\sigma^\omega_{-t}$ \cite{HaagerupHanche-Olsen1984}. 

What is more, the intrinsic dynamics in von Neumann algebras is further exemplified in the study of statistical mechanics in a background-independent setting \cite{Rovelli1993a,Rovelli1993b}. Here, $\omega \in \SN$ is understood as a state in thermodynamic equilibrium. In the setting of quantum statistical mechanics such states are characterised by the KMS condition \cite{HaagHugenholtzWinnink1967}. It is a remarkable fact that $\sigma^\omega$ satisfies the KMS condition for every faithful normal state $\omega \in \SN$ \cite{Takesaki1970}. In contrast, no analogue of this condition holds for Jordan algebras \cite{HaagerupHanche-Olsen1984}. In effect, time orientations in von Neumann algebras allow to interpret time from a thermodynamical standpoint, encoded in a state of thermodynamic equilibrium \cite{ConnesRovelli1994}.

The crucial role played by time orientations (equivalently, dynamical correspondences) in von Neumann algebras and in Thm.~\ref{thm: entanglement and time} is hardly coincidental. In particular, it is tempting to `explain' the intrinsic dynamics and (thermodynamic) origin of the arrow of time more fundamentally in terms of the entanglement structure of a given faithful normal state. %This suggests that underlying the statistical nature of KMS states might be their quantum entanglement.
To this end, one would like to generalise %and internalise
Thm.~\ref{thm: entanglement and time} to the setting of general von Neumann algebras.
%\footnote{Note that under this generalisation a completely positive map corresponds with a positive operator via its Radon-Nikodym derivative \cite{Belavkin1986,FrembsDoering2021}. (This foreshadows the relativity of time orientations between states. [- to be studied elsewhere?!])}
%In particular, it would be interesting to relate Connes' cocylce condition with Thm.~\ref{thm: entanglement and time}, e.g. using Thm.~7 in \cite{FrembsDoering2022b}.
%flag: just AlfsenShultz1998?
%Moreover, it would be interesting to generalise Thm.~\ref{thm: entanglement and time} to many parties. %for instance, in the setting of causal nets in algebraic quantum field theory.
We leave this and similar considerations for future work.

%We close with a sketch of our results can be extended to this setting by considering tensor product decompositions of von Neumann algebras. Let $\omega \in \mc{S}(\cN)$ be faithful, then $\omega$ is represented as a cyclic and separating vector with respect to the GNS construction of $\omega$. As before, $\omega$ defines a dynamics for the whole system in terms of the one-parameter groups $\chi_\omega$. Next, consider a decomposition $\cN = \cN_1 \otimes \cN_2$ (in terms of the spatial tensor product) with respect to the Hilbert space of the GNS representation of $\cN$ and $\omega$.  With respect to this decomposition
%%we find $S_\omega$ splits $S_\omega = S_{\omega,1} \otimes S_{\omega,2}$, thus defining
%$\omega \in \SN$ therefore defines a relative orientation between the subsystems. (Interpreting $\omega \in \SN$ as a statistical background, which exhibits negligible interactions with some fields defined on it \cite{Rovelli1993a,Rovelli1993b}], one finds that the thermodynamic time defined in the background appears locally as entanglement.)}\\

%% file: Discussion/Conclusion.tex
We found a necessary and sufficient criterion for bipartite entanglement using Stinespring dilations in Thm.~\ref{thm: dilated PPT criterion}. The latter adopts a clear physical meaning in terms of a compatibility condition with respect to time orientations (Def.~\ref{def: time-oriented state}) on the respective local observable algebras in Thm.~\ref{thm: entanglement and time}. Moreover, we highlighted the key role time orientations %(equivalently, dynamical correspondences)
play within the broader picture of the intrinsic flow of time in von Neumann algebras. 

More explicitly, our results are motivated from and bear close resemblance with the PPT criterion \cite{Peres1996,Horodeckisz1996}. As such, it would be interesting to study the practical relevance of Thm.~\ref{thm: dilated PPT criterion}. For example, it seems possible that existing results on marginal extension problems, e.g. \cite{DohertyParriloSpedalieri2002,DohertyParriloSpedalierei2004}, can be strengthened using Cor.~\ref{cor: commutative extension problem}.\\

%Finally, Prop.~\ref{prop: PPT trivial on decomposable states} and Thm.~\ref{thm: dilated PPT criterion} show that the PPT criterion is concerned with the difference between decomposable and completely positive maps. Both are special cases of positive maps. A complete classification of the latter has proven to be hard \cite{Majewski2020}. Note that the lack of positivity has also been at the heart of the problem of the classification of quantum from non-signalling bipartite correlations \cite{Frembs2022}. In contrast, the difference between positivity and complete positivity plays a crucial role in the study of Bell nonlocality, compared to quantum entanglement. The latter obtains a sharp resource-theoretic characterisation in \cite{Buscemi2012b}, while Bell nonlocality requires a relaxation to positive maps \cite{FrembsBuscemi2022}.\\